\documentclass[lettersize,journal]{IEEEtran}
\usepackage{amsmath,amsfonts}
\usepackage{algorithmic}
\usepackage{algorithm}
\usepackage{array}
\usepackage[caption=false,font=normalsize,labelfont=sf,textfont=sf]{subfig}
\usepackage{textcomp}
\usepackage{stfloats}
\usepackage{url}
\usepackage{verbatim}
\usepackage{graphicx}
\usepackage{cite}
\usepackage{amsthm,pdftexcmds}
\usepackage{amssymb}
\usepackage{fancyvrb}
\usepackage{wasysym}
\usepackage{tikz}

\RequirePackage{xcolor}
\RequirePackage{listings}
\usepackage{matlab-prettifier}

\usepackage{lipsum}
\usepackage[d]{esvect}
\usepackage{bm}

\usepackage{accents}
\usepackage{verbatim}

\usepackage{siunitx}

\usepackage{fvextra}

\usepackage{color,soul}

\usepackage{mathtools}

\newtheorem{theorem}{Theorem}
\newtheorem{corollary}{Corollary}[theorem]
\newtheorem{lemma}[theorem]{Lemma}
\newtheorem{remark}{Remark}

\theoremstyle{definition}
\newtheorem{definition}{Definition}

\newtheorem{example}{Example}

\theoremstyle{remark}

\newcommand{\lec}{\underset{\overset{c}{}}{\le}}
\newcommand{\gec}{\underset{\overset{c}{}}{\ge}}

\newcommand{\eqc}{\underset{\overset{c}{}}{=}}

\newcommand{\avec}[1]{\accentset{\rightharpoonup}{#1}}

\begin{document}
\title{Algorithmic Information Forecastability 
}

\author{Glauco Amigo,~\IEEEmembership{Member,~IEEE,}
	Daniel Andr\'{e}s D\'{i}az-Pach\'{o}n,~\IEEEmembership{Member,~IEEE,}\\
	Robert J. Marks,~\IEEEmembership{Life Fellow,~IEEE,}
		Charles Baylis,~\IEEEmembership{Senior Member,~IEEE,}
\thanks{This work was supported in part by Discovery Institute's Walter Bradley Center for Natural and Artificial Intelligence.}
}

\markboth{IEEE Transactions}%
{Amigo \MakeLowercase{\textit{et al.}}: Algorithmic Information Forecastability}


\maketitle

\begin{abstract}
The outcome of all time series cannot be forecast, e.g. the flipping of a fair coin. Others, like the repeated $\{01\}$ sequence $\{ 010101\cdots \}$ can be forecast exactly.   Algorithmic information theory can provide a measure of {\em forecastability} that lies between these extremes. The degree of forecastability is a function of only the data.
For prediction (or classification) of labeled data, we propose three categories for forecastability: oracle forecastability for predictions that are always exact, precise forecastability for errors up to a bound, and probabilistic forecastability for any other predictions. Examples are given in each case.
\end{abstract}


\section{Introduction}

In 2010, artificial intelligence pioneer Marvin Minsky noted the importance of algorithmic
information theory applied to the field of machine learning.
\begin{quotation} \noindent
	``It seems to me that the most important discovery since G\"{o}del was the discovery by Chaitin, Solomonov and Kolmogorov of a concept called algorithmic probability, which is a fundamental new theory of how to make predictions given a collection of experiences. And, this is a beautiful theory, everybody should learn it\dots \ But it's got one problem, which is that you can’t actually calculate what this theory predicts because it’s too hard, requires an infinite amount of work. However, it should be possible to make practical approximations to the Chaitin-Kolmogorov-Solomonov theory that would make better predictions than anything we have today. And everybody should learn all about that and spend the rest of their lives working on it.''  \\
	\hspace*{\fill} ---Marvin Minsky \cite{minski}
\end{quotation}
Both Minsky's enthusiasm and cautions are spot on. Exact computation of the algorithmic complexity of any object is unknowable. But it can be bounded, for example, using developed procedures used in the field of {\em minimum description length} \cite{grunwald2007minimum}. The algorithmic information theory paradigm provides valuable insight which we now apply to time series forecastability.

Simply stated, prediction machine learning requires future data to be forecast from past data. For instance, machine intelligence requires images used in training in the past to be in some sense in the same class as images to be classified in the future. We define data to be {\em forecast ergodic} \cite{Gilder} if it can be predicted from the labeled training data and the current input. {\em Algorithmic information forcastability} (AIF), or forecast ergodicity, is a measure of the ability to forecast future events from data in the past. We model this bound from the algorithmic complexity of the available data. 

We use the word ergodic as it is generically defined as ``of or relating to a process in which every sequence or sizable sample is equally representative of the whole'' \cite{MerriamWebster}.  In time series, ergodicity concerns the estimation of parameters from a single observation \cite{alaoglu,cornfeld,papoulis}.  
The time series generated by the repeated flipping of a fair coin, 1 for heads and 0 for tails, is {\em mean ergodic} since the mean of $\frac{1}{2}$ can be estimated from a single observation of a time series of flipping a coins. A series is forecast ergodic if, given enough data, the future can be specified or estimated from past data. 
Coin flipping is not forecast ergodic since the exact value of future coin flip cannot be determined by previous coin flips. 

Traditional computational learning theory \cite{angluin,anthony} such as Valiant's probabilistic almost correct (PAC) approach \cite{valiant1984,valiant2013}, is based on probability and models.  We propose the application of Kolmogorov-Chaitin-Solomonov (KCS) complexity\footnote{
	a.k.a. {\em Kolmogorov complexity} \cite{cover,vitanyi}. } 
\cite{marks2017} and algorithmic information theory \cite{tadaki2019} to model forecast ergodicity. Unlike probabilistic models and Shannon information theory, algorithmic information theory is based on data structure. The forecast ergodicity we propose is model-free and determined only by data. 

The motivation behind forecast ergodicity is to show the computational boundaries to prediction of future outputs from available  data and inputs. Predictions can be done after  problem reduction to inputs and outputs in terms of computation (see Figure \ref{fig:forerggraph}). The goal is to extract the structure (the algorithmic information) in the available data  to forecast future outputs  by applying this structure to new inputs. Forecast accuracy is determined by the algorithmic complexity of the available data. 
For forecast ergodic data, if some structure exist that relates inputs with outputs in the training data, then the shortest program that generates the training data necessarily exploits the existing structure for better compression. If the existing structure applies to the future, then the problem is forecast ergodic.

In the next section we introduce terminology required for the paper. The proposed model of forecasting in Section \ref{sec:FEAI} builds on Solomonov's theory of prediction \cite{SOLOMONOFF1,SOLOMONOFF2}, which indicates that ``shortest programs almost always lead to optimal prediction'' \cite{vitanyi}.
Since it is based on  Kolmogorov complexity, forecast ergodicity can only be bounded \cite{RISSANEN1978465,Bhat2010BIC}.

\begin{figure}
	\centering
	\includegraphics[width=1\linewidth]{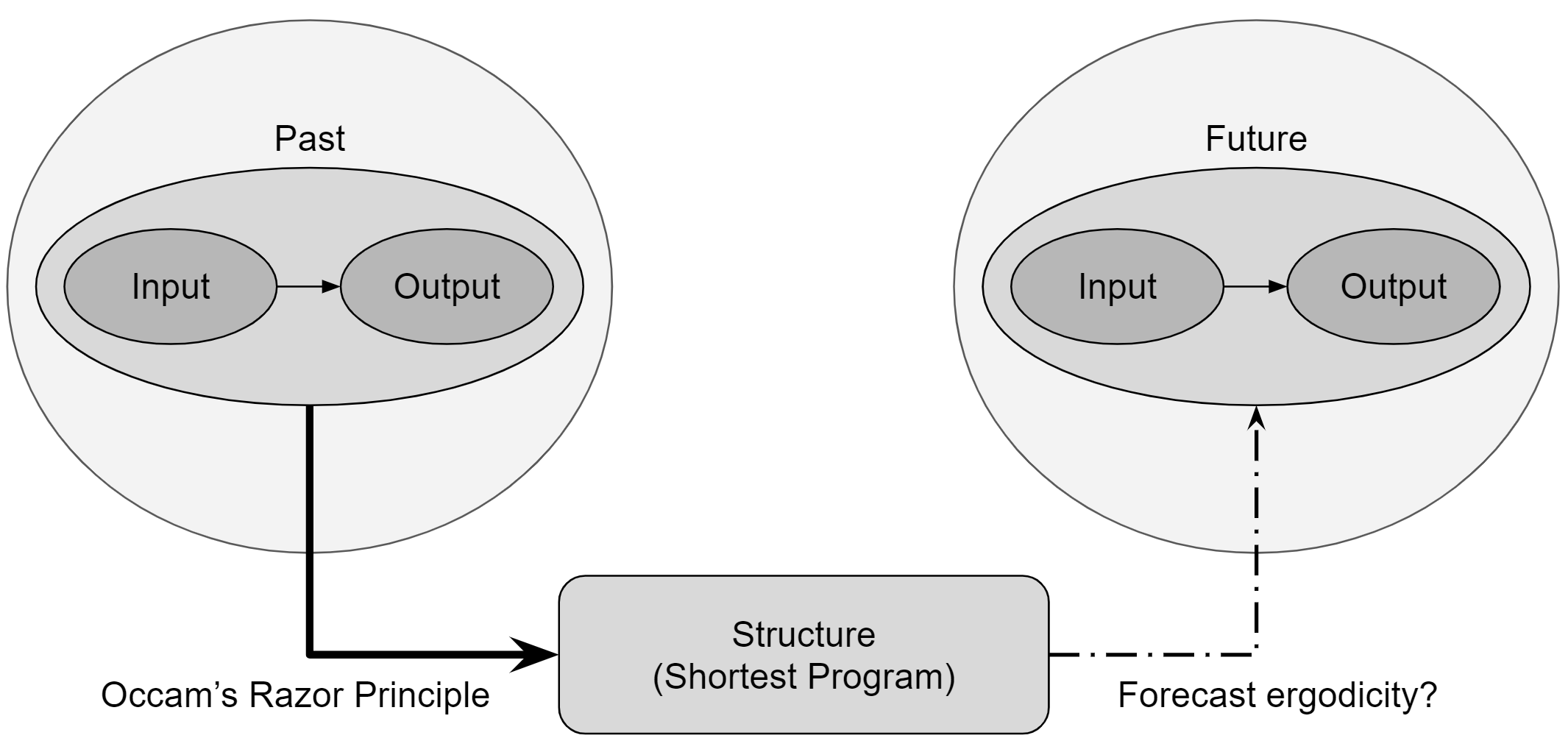}
	\caption{Illustration of problem reduction to inputs and outputs in forecasting. }
	\label{fig:forerggraph}
\end{figure}


\section{Background}
For continuity a short introduction to KCS complexity is appropriate. 

\subsection{KCS Complexity}
KCS complexity considers a computer program $p$, a universal computer $\mathcal{U}$ and the corresponding computer output  $x=\mathcal{U}(p)$.  Let the length of program $p$, in bits, be $\ell(p)$. There are many programs that can generate a given output $x$. The \textit{KCS complexity}, $K(x)$, is the length of the shortest binary program $p$ that generates the object $x$ \cite{cover, vitanyi, ASC, structureFunctions}:
\begin{equation}\label{eq:kolmogorov_complexity}
	K_\mathcal{U}(x) = \min_{p}\left\{\ell(p):\mathcal{U}(p)=x\right\}.
\end{equation}

When additional information $y$ is provided as an input to the universal computer $\mathcal{U}$ along with the program $p$ to generate the binary string $x$, the \textit{conditional KCS complexity} is defined as: 
\begin{equation}\label{eq:conditional_kolmogorov_complexity}
	K_\mathcal{U}(x|y) = \min_{p}\left\{\ell(p):\mathcal{U}(p,y)=x\right\}.
\end{equation}

For every program $p$, for a universal computer $\mathcal{U}$, the program $p$ can be translated to any other universal computer $\mathcal{A}$. The maximum cost of the translation is the constant $c
$ 
equal to the length of the shortest program that translates from $\mathcal{U}$ to $\mathcal{A}$ \cite{cover}. So, for any pair of universal computers $\mathcal{U}$ and $\mathcal{A}$
$$ K_{\mathcal{U}}(x) \le K_{\mathcal{A}}(x) + c. $$
The constant $c$ does not depend on $p$. So it is possible to use any universal computer $\mathcal{U}$ as a fixed reference to measure the algorithmic complexity by adding the constant. 
More generally,
\begin{equation}
	\mid	K_{\mathcal{U}}(x) -K_{\mathcal{A}}(x) \mid \leq c 
	\label{x201229}
\end{equation}
where the constant $c$ is independent of $p$.  	
The bound in \eqref{x201229} can be written more simply as 
$$ K_{\mathcal{U}}(x) \eqc K_{\mathcal{A}}(x). $$
This equality allows us to drop reference to the specific computer  $\mathcal{U}$ in future discussions about KCS complexity. 

\begin{example}\label{ex:simple_example}
	Consider the  program that prints a concatenation $x$ of the binary numbers from 0 to  $n$ (e.g. the binary Champernowne sequence \cite{calude,champernowne} for the first $n +1$ non negative integers):		
	\begin{Verbatim}[breaklines=true, breakanywhere=true]
PROGRAM 1:
print('011011100101110111100010011010101111001101111011111000010001100101001110100...')
	\end{Verbatim}
	Another program that prints the same binary string  $x$ is:
	\begin{Verbatim}[breaklines=true, breakanywhere=true]
PROGRAM 2:
for i in range(n+1):
  print('{0:b}'.format(i), end='')
	\end{Verbatim}
	The length of  program 1 grows as $O\left(n\log n \right)$, while the length of  program 2 grows as $O(\log n)$.\footnote{Throughout this paper the base of the logarithm is 2.}
	So, for large $n$, program 2 is always shorter than program 1. where  $\log n$ relates to the length of the bit stream in the second program. 
\end{example}
From this example, given a concatenation $x$ of binary numbers from 0 to $n$, the KCS complexity of $x$ is bounded by 
\[
K\left(x \mid  \ell(p)\right) \leq \log n +c 
\]

\section{Forecastability}\label{sec:FEAI}
This section contains the core definitions on algorithmic information forecastability and some related analysis illustrated with simple examples. More detailed examples are presented in Section \ref{sec:examples}. 

The data from the past is the training data. Both the past data and the future data include inputs and outputs. In the definitions, the available data  includes the  past data and some inputs of the future data. The predictions are the outputs of the future. 

Let $\mathcal{F}$ be the countable set of all finite binary strings.
Let $\avec{X}_M$ be a $M$ vector such that each of its components belongs to  $\mathcal{F}$,   where  $M \in \mathbb{N} \cup \{\infty \} $.  For $N \le M$, let $\avec{X}_N$ be the first $N$ components of input data  $\avec{X}_M$. Define analogously the output data $\avec{Y}_M$ and $\avec{Y}_N$.
The length of  $x_m \in \avec{X}_M$ can be different than the length of  $y_m \in \avec{Y}_M$.  
Table \ref{tab:outline-nomenclature-1} outlines these elements. 

The pair $\left(\avec{X}_N, \avec{Y}_N\right)$ represents the training data, which can be viewed as  occuring in the past. The goal of any machine intelligence application is to forecast the future outputs $\avec{Y}_M \setminus \avec{Y}_N = \left\{y_{N+1}, y_{N+2}, \dots, y_M\right\}$ given  $\left(\avec{X}_M, \avec{Y}_N\right)$, where for two sets $A$ and $B$, $A \setminus B = A\overline{B} $ is the set difference between $A$ and $B$.

\begin{table}[htbp]
	\centering
	\caption{Outline of training and future.}
	\label{tab:outline-nomenclature-1}
	\setlength\extrarowheight{5pt}
	\begin{tabular}{c|c|c|}
		\cline{2-3}
		\multicolumn{1}{l|}{}              & Past (Training Data) & Future                            \\ \hline
		\multicolumn{1}{|c|}{Input}        & $\avec{X}_N$         & $\avec{X}_M \setminus \avec{X}_N$ \\ \cline{2-3} 
		\multicolumn{1}{|c|}{$\avec{X}_M$} & $x_1, \dots, x_N$    & $x_{N+1}, \dots, x_M$             \\ \hline
		\multicolumn{1}{|c|}{Output}       & $\avec{Y}_N$         & $\avec{Y}_M \setminus \avec{Y}_N$ \\ \cline{2-3} 
		\multicolumn{1}{|c|}{$\avec{Y}_M$} & $y_1, \dots, y_N$    & $y_{N+1}, \dots, y_M$             \\ \hline
	\end{tabular}
\end{table}

\subsection{Oracle Forecastability}
We first look at the case where predictions are always exact. That is, given the training data and the future input, the future output is perfectly predicted. This is achieved by means of the structure of the training data. We call these predictions to be {\em oracle forecastable}, because the shortest program that compresses the training data can be applied to the future input and print the future output with perfect precision, working as a flawless oracle. In the general case, the structure of the training data will compress the data from the past in a way that will include subroutines for the different classes of inputs. Within each class, the relation between input and output is the same, so, the structure compresses that class with a subroutine that implements the relation with an identification of the set of inputs that the relation applies to. 

\begin{definition}\label{def:ergodic-ai} 
	The pair of vectors $\left(\avec{X}_M, \avec{Y}_M\right)$ is \textit{oracle forecastable} (OF) for $N$ if,
	for all $ m \in \{N,\dots,M\}$,
	\begin{equation}\label{eq:ergodic-ai}
		K\left( 
		y_m \middle\vert \avec{X}_N, \avec{Y}_N, x_m 
		\right) \eqc 0.
	\end{equation}
\end{definition}

The equality in  \eqref{eq:ergodic-ai}, as in subsequent definitions, is up to a constant. Here, this constant is the length of the \textit{forecasting program} responsible of doing the prediction.
The forecasting program  pipes the future input to the  subroutine  that implements the structure from the training data and prints the output.

The pair $\left(\avec{X}_M, \avec{Y}_M\right)$ is trivially OF for $N=M$.

According to Definition \ref{def:ergodic-ai}, OF for $N$ requires the conditional KCS complexity of every future ($m>N$) output,  given the future input and the training data, to be 0. Therefore, under OF perfect accuracy is achieved.

\begin{remark}
	OF does not suggest how to obtain perfect accuracy.
\end{remark}

Here are two elaborations.
\begin{enumerate} 
	\item For $N' \in \{N+1, \dots, M\}$, let $\avec{X}_{N'} = \left(\avec{X}_N, x_{N+1}, \cdots, x_{N'} \right)$. A  particular case of an OF for $N$ sequence is
	\begin{equation}\label{eq:ergodic-ai-w}
		K\left( y_{N'} \middle\vert \avec{Y}_N, \avec{X}_{N'} \right) \eqc 0.
	\end{equation}
	Thus, the future input in \eqref{eq:ergodic-ai} is made only of $ x_m$, whereas, in \eqref{eq:ergodic-ai-w}
	the future input includes $x_{N+1}, x_{N+2},\ldots, x_ {N'}$.
	To present this particular case more clearly, we can go from \eqref{eq:ergodic-ai-w} to \eqref{eq:ergodic-ai}:
	\begin{align*}
		K\left( y_{N'} \middle\vert \avec{Y}_N, \avec{X}_{N'} \right) &= 
		K\left( y_{N'} \middle\vert \avec{Y}_N, \avec{X}_N, x_{N+1}, \dots, x_{N'}  \right)\\
		&= K\left( y_{N'} \middle\vert \avec{X}_N, \avec{Y}_N, x_{N'}^*  \right),
	\end{align*}
	where $ x_{N'}^* =\left(x_{N+1}, \dots, x_{N'}\right)$. That is, the vector of  inputs between $N+1$ and $N'$ can be  considered as the future input in \eqref{eq:ergodic-ai}. This addresses the fact that  algorithms that change over time are included in Definition \ref{def:ergodic-ai}.  Algorithmic information forecastability thus applies to both  static algorithms and algorithms that change over time. 
	
	\item
	By Definition \ref{def:ergodic-ai},
	\begin{equation}\label{eq:ergodic-ai-inf}
		\text{OF} \text{ for } N \Rightarrow  K\left(\avec{Y}_M \middle\vert \avec{Y}_N, \avec{X}_M\right) \eqc 0.
	\end{equation}
	The converse is true \textit{for some $\mathcal{M} \in \{N, \dots, M-1\}$}:
	\begin{equation}\label{eq:ergodic-ai-inf-rev}
		\text{OF for some } \mathcal{M}  \Leftarrow K\left(\avec{Y}_M \middle\vert \avec{Y}_N, \avec{X}_M\right) \eqc 0.
	\end{equation}
	However, depending on the pair of vectors $\left(\avec{X}_M, \avec{Y}_M\right)$ being finite or infinite the significance of the implication in \eqref{eq:ergodic-ai-inf-rev} is different.
	\begin{lemma}
		For  $\left(\avec{X}_M, \avec{Y}_M\right)$  \textbf{finite} and $N<M$, if $K\left(\avec{Y}_M \middle\vert \avec{Y}_N, \avec{X}_M\right) \eqc 0$, then $\left(\avec{X}_M, \avec{Y}_M\right)$ is OFE at least for $N= M-1$.
	\end{lemma}
	\begin{proof}
		The vector $\avec{X}_M$ might contain  information  which is not  in  $\avec{X}_{N}$ when $N<M$. But if $N = M-1$, then $\avec{X}_M \setminus \avec{X}_N = x_{N+1} = x_M$, which corresponds to the input $x_m$ in Definition \ref{def:ergodic-ai}. In this case the information only comes from  one element that satisfies the definition of OF.
	\end{proof}
	\begin{example}
		Electric load forecasting consist of predicting the demand for electricity
		based on historical data and other inputs such as weather forecast.
		Electric load forecasting  is not OF for $N<M$ because the output cannot be forecast exactly \cite{electric_load_forecast}.
		If data is collected on a daily basis, the forecasting inputs can include at most the electric load of the previous day. Therefore, were the input of tomorrow available, the forecasting of today would be trivial.
		Alas, the input of tomorrow is not available today.
		Additionally, this example illustrates that $\left(\avec{X}_M, \avec{Y}_M\right)$ is not OF for any $N<M$   when the output is random, according to Definition \ref{def:ergodic-ai}.
	\end{example}

	The next lemma says that if the conditional Kolmogorov complexity of the output given the input and the training data goes to zero as $M \rightarrow \infty$, then the limit is actually hit for some finite $\mathcal M$. The implication is that under this condition, the pair $\left(\avec{X}_M, \avec{Y}_M\right)$ is OF for all $m > \mathcal{M}$.
	\begin{lemma}\label{lem:limitHit}
		If 		
		\begin{equation}\label{eq:ergodicity-proof}
			\lim_{M \rightarrow \infty} K\left( \avec{Y}_{M} \middle\vert \avec{Y}_{N}, \avec{X}_{M} \right) \eqc 0,
		\end{equation}
		then there exists $\mathcal{M}$ such that for all $ m > \mathcal{M}$, 
		$$ K \left( y_m \middle\vert \avec{Y}_N, \avec{X}_m \right) \eqc 0. $$
	\end{lemma}
	\begin{proof}
		Assume that for all $\mathcal{M}$ there is an $m>\mathcal M$ such that $ K \left( y_m \middle\vert \avec{Y}_N, \avec{X}_m \right) $ is not zero. Since Kolmogorov complexity is always a non-negative integer, this implies that $\lim_{M \rightarrow \infty} K\left( \avec{Y}_{M} \middle\vert \avec{Y}_{N}, \avec{X}_{M} \right) \gec 1$, which contradicts~\eqref{eq:ergodicity-proof}.
	\end{proof}
	\begin{corollary}\label{cor:limitHit}
		Choose $\mathcal{M}$ as introduced in Lemma \ref{lem:limitHit}  and substitute $\avec{Y}_N$  by $\avec{Y}_\mathcal{M} = \{\avec{Y}_N, y_{N+1},\dots,y_\mathcal{M}\}$, then there exists $\mathcal{M}$ such that for all $m > \mathcal{M}$,
		$$ 	K \left( y_m \middle\vert \avec{Y}_\mathcal{M}, \avec{X}_m \right) \eqc 0 $$
		which is equivalent to the particular case  of OF for $\mathcal{M}$  in  \eqref{eq:ergodic-ai-w}.
	\end{corollary}
\end{enumerate}

The next two lemmas address some applications to Markov chains.
\begin{lemma}\label{lem:absMarkov}
	Let $\left\{\avec{X}_M\right\}$  be an absorbing Markov chain with state space $\Omega$. Let 
	\begin{align*}
		f: \Omega  &\to \mathcal{F} \\
		x_m &\mapsto y_m.
	\end{align*}
	Then the pair $\left(\avec{X}_M, \avec{Y}_M\right)$  is OF for some finite $N$.
\end{lemma}
\begin{proof}
	In an absorbing Markov chain, after being in any of the states of the chain, an absorbing state is reached in a finite number of steps. As $M \rightarrow \infty$, an absorbing state will be reached at some time almost surely. Let $a$ be the first time an absorbing state is reached. Then for all $m \ge a$, $x_m = x_a$ and $ y_m = y_a$, and the result follows. 
\end{proof}

\begin{lemma}
	Let $\left\{\avec{X}_M\right\}$  be an irreducible finite Markov chain with state space $\Omega$. Let 
	\begin{align*}
		f: \Omega  &\to \mathcal{F} \\
		x_m &\mapsto y_m.
	\end{align*}
	Then the pair $\left(\avec{X}_M, \avec{Y}_M\right)$  is OF for some finite $N$.
\end{lemma}
\begin{proof}
	An irreducible finite Markov chain is also recurrent, that is, every state will appear in $\left\{\avec{X}_M\right\}$ infinitely often. Choosing $N$ large enough, the training data will include every possible input and output of $f$ arbitrarily many times, and the result follows.
\end{proof}

\subsection{Precise Forcastability}
Precision is defined for each application. For example, by an absolute error distance, by a percentage, by a Hamming distance, or by design constraints. Notice that, in general,  an output is precise if it is in a predefined set of possible outcomes. To formally encompass every possible precision constraint,  the following formulation is used.

Let $B(y_m, \epsilon_m) = \left\{y^*_m : d(y_m,y^*_m) < \epsilon_m \right\} $ be a ball with center $y_m$ and radius $\epsilon_m$ under certain metric $d$, defining the precision.

\begin{definition} \label{def:precise-ergodic}
	The pair $\left(\avec{X}_M, \avec{Y}_M\right)$ is \textit{precise forecastable} (PF)  up to $\epsilon$ for $N$, if, for all  $m \in \{ N, \dots, M\}$, there exists $ y^*_m \in B(y_m, \epsilon_m)$ such that 
	\begin{equation}\label{eq:precise-ergodic}
		K\left( y^*_m \middle\vert \avec{X}_N, \avec{Y}_N, x_m \right) \eqc 0 ,
	\end{equation}
	where $\epsilon = [\epsilon_1,\ldots,\epsilon_M]$ is the level of precision.
\end{definition}

The meaning underlying this definition is parallel to the previous case of OF. We are using the structure that is already present in the training data and in the current input to determine the desired output. And this is achieved by means of the algorithmic information of the available data, so that the Kolmogorov complexity of the output, given the available data, is zero. The only thing that changes in the definition of PF is that we introduce a tolerance for the precision of the forecast.

We have formalized the precision in the broadest possible way for PF to be applicable in the general case. More particularly, the following lemma elaborates on one of the most common and intuitive conceptions of precision, that is, when the error is bounded by the absolute value distance. In this case, there is a relation between the magnitude of the absolute value error and the KCS complexity of the error but this relation is not one of equivalence.
\begin{lemma}\label{lem:ball}
	For the ball 
	\begin{equation}\label{eq:pre_epsilon}
		B(y_m, \epsilon_m) = \left\{  y^*_m : \left| y_m -  y^*_m \right| < \epsilon_m \right\},
	\end{equation}
	PF implies
	\begin{equation}\label{eq:ce}
		K\left( y_m \middle\vert \avec{X}_N, \avec{Y}_N, x_m \right) \le C_\epsilon ,
	\end{equation}
	for all $m \in \{ N, \dots, M\}$,
	where $C_\epsilon$ is a constant that bounds the KCS complexity of the error.
\end{lemma}
\begin{proof}
	We have
	\begin{equation}\label{asdf}
		K\left( y_m \middle\vert \avec{X}_N, \avec{Y}_N, x_m \right) \eqc K\left( y_m \middle\vert y_m^* \right) \lec C_\epsilon ,
	\end{equation}
	where equality is up to a constant that accounts for the difference between the length of the forecasting program in  $K\left( y_m \middle\vert \avec{X}_N, \avec{Y}_N, x_m \right)$ and the length of the program that performs an arithmetic operation in $K\left( y_m \middle\vert y_m^* \right)$. The latter also accounts for a constant in the inequality. Let
	\begin{equation}\label{eq:cer}
		C_\epsilon =  \max_{m \in [N,\dots,M]} K\left( y_m \middle\vert y_m^* \right) \lec  \lg^*(\epsilon_m)
	\end{equation}
	where $\lg^*(\epsilon_m)$ is the length of a self-delimiting encoding of the $\epsilon_m$ in bits (see Appendix \ref{apx:log_star}). The inequality is up to a constant that includes the sign or direction of $y_m - y_m^*$.
\end{proof}
The previous lemma shows that, for the ball in \eqref{eq:pre_epsilon}, if the error is bounded, then the KCS complexity of the error is bounded as well. 
However, the converse is not true.
To explain it, there are instances of low KCS complexity, such as the instruction in pseudocode ``\texttt{flip\_all\_bits}", that can be short enough to satisfy  \eqref{eq:ce} while not fitting the definition of PFE. 
In general, a bound for the KCS complexity of the error does not bound the error itself.

In a different example, the precision can consist of a Hamming distance bound, where the Hamming distance is defined as $$H\left(y,y^*\right) = \sum_{i=1}^{\ell} \left[y_i \ne y_i^*\right]$$ 
for two binary strings $y = (y_1,\dots,y_\ell), y^* = (y_1^*,\dots,y_\ell^*)$. Consider, for instance, images with black and white pixels, for which black pixels take the binary value of 0 and white pixels take value 1. Images are  visualized in 2 dimensions, but they can be converted to one dimensional arrays for the theoretical analysis. Then
\begin{equation}\label{eq:Hamming}
	B(y_m, \epsilon_m)  = \left\{  y^*_m : \text{H}\left( y_m ,  y^*_m \right) \le h \right\} 
\end{equation}
where  $h$ is the bound for the precision.

In another example, the distance $d$ can also be a Boolean attribute. For instance, assessing if a solution satisfies a series of engineering requirements, with margin to favor selected conflicting goals.

The previous examples illustrate the general formulation of PF and its applicability to different metrics. Note the relation between OF and PF:

\begin{remark}\label{rem:OFE_PFE}
	OF for $N$, in \eqref{eq:ergodic-ai}, is a particular case of PF for $N$, in \eqref{eq:precise-ergodic}, where the precision is maximally constrained with $B(y_m, \epsilon_m)  = \left\{ y_m  \right\} $ for all $m$.
\end{remark}

So, it follows that  OF is a subset of PF, as every  pair $\left(\avec{X}_M, \avec{Y}_M\right)$ that is OF for $N$ is also PF for $N$.

\subsection{Probabilistic Forecastability}
Probabilistic forecastability applies when a subset of elements in a pair $\left(\avec{X}_M, \avec{Y}_M\right)$ is PF. In plain language this means that there is certain probability that the predictions will be precise, which is formalized in the following definition.
\begin{definition} \label{def:probabilistic-ergodic}
	The pair $\left(\avec{X}_M, \avec{Y}_M\right)$ is \textit{probabilistic forecastable} (PrF) for $N$, if for a fraction $P$  of the elements $\left( x_m, y_m \right) $, $m \in \{ N, \dots, M\}$,  there exists $ y^*_m \in B(y_m, \epsilon_m)$ such that 
	\begin{equation}\label{eq:probabilistic-ergodic}
		K\left( y^*_m \middle\vert \avec{X}_N, \avec{Y}_N, x_m \right) \eqc 0.
	\end{equation}
\end{definition}

This definition follows the same logic as in OF and PF. That is, given the structure of the available data, the prediction does not need any additional algorithmic information (this is why the expression equals zero in \eqref{eq:probabilistic-ergodic}). Everything that is needed is a program that identifies the class of the current input, locates the subroutine for that class in the structure, and obtains the corresponding output (the length of this program is included in the constant $c$ under the equal sign in \eqref{eq:probabilistic-ergodic}). What differentiates PrF for $N$ from PF for $N$ is that the probability of achieving precision is not necessarily 1.

The pair  $\left(\avec{X}_M, \avec{Y}_M\right)$ is trivially PrF for $P=0$.

PrF is broad enough to encompass any amount of algorithmic information from the available data with capacity to forecast some future data. The probability $P$ accounts for the ratio of the  future elements $\avec{Y}_M \setminus \avec{Y}_N$ that are forecastable to a given degree of precision. 

Depending on the problem, only one bit of information or larger amounts of information can be necessary. To visualize it with a simple example, a bride and bridegroom may only need one bit of information to know if it is going to rain or not on the day of their wedding to decide if they celebrate it indoors or outdoors. Instead, a farmer needs  more information about the amount of rain. One bit can still be significant for him, but he would need at least 3 categories: (a) not enough rain, (b) just the right amount of rain and (c) too much rain; that is more than one bit of information. The farmer needs  more details about the rain than the couple. And the farmer  needs that information for each day of the year with high probability that the precipitation will fit the forecast. However, in the case of the wedding, they only need one bit of information, but they would need it more accurately if possible. 

Regarding the amount of forecastable elements that non-trivial PrF requires:
\begin{itemize}
	\item When $\left( \avec{X}_M, \avec{Y}_M \right)$ is  finite, PrF for $N$ only requires some  $y_m$, $m \in \{ N, \dots, M\}$,  to be forecastable at a particular precision. Even if there is only one forecastable pair $\left( x_m, y_m \right) $, that  is sufficient to make $P>0$
	\item When $\left( \avec{X}_M, \avec{Y}_M \right)$ is  infinite, non-trivial PrF for $N$ requires  a fraction $P>0$ of the  elements $y_m$ from $\avec{Y}_M \setminus  \avec{Y}_N$  to be forecastable at a particular precision. No matter how small the fraction $P$ is, the amount of such forecastable elements must be infinite. Otherwise, a finite amount of elements in an infinite series is negligible, given that no information about their location within the series is provided. If their location was known, it would be possible to truncate an infinite $\left( \avec{X}_M, \avec{Y}_M \right)$ pair into a finite one
\end{itemize}

\begin{remark}\label{rem:PFE_PrFE}
	PF is a particular case of PrF where $P=1$.
\end{remark}
The previous remark shows that PF is a subset of PrF, as every  pair $\left(\avec{X}_M, \avec{Y}_M\right)$ that is PF for $N$ is also PrF for $N$ with $P=1$. The conditions for PrF and PF (Definition \ref{def:probabilistic-ergodic} and \ref{def:precise-ergodic}, respectively) are similar, but while PrF only requires a fraction $P$ of the elements in  $\avec{Y}_M \setminus \avec{Y}_N$ to satisfy those conditions, PF requires all elements in $\avec{Y}_M \setminus \avec{Y}_N$ to satisfy those conditions.

\subsection{Forecast Ergodicity}
In this section we propose the concept \textit{forecast ergodicity} as a measure of the ability to forecast future events from data in the past. We model this bound by the algorithmic complexity of the available data.
\begin{definition}\label{def:ergodicity-ai} 
	The \textit{forecast ergodicity (FE)}, $\mathcal{E}\left(\avec{X}_M, \avec{Y}_M \right)$, of a pair   $\left(\avec{X}_M, \avec{Y}_M \right)$ is defined as a measure of the algorithmic information in the training data that is useful to forecast the future. The FE is bounded by:
	\begin{equation} \label{eq:FE}
		\mathcal{E}\left(\avec{X}_M, \avec{Y}_M \right) \le
		K\left( \avec{Y}_N \middle\vert \avec{X}_N \right).
	\end{equation}

\end{definition}

Note that the inequality in \eqref{eq:FE} requires either all of the structure provided by the training data is needed to forecast the future, or a shorter program could do. That is,  a part of the structure in the training data may not be necessary to forecast the future.

Let 
\begin{itemize}
	\item $\mathcal{E}_{\text{OF}}$ be the set of all OF pairs  $\left(\avec{X}_M, \avec{Y}_M \right)$ for $N$
	\item $\mathcal{E}_{\text{PF}}$ be the set of all PF pairs  $\left(\avec{X}_M, \avec{Y}_M \right)$  for $N$ and precision  $\epsilon$
	\item $\mathcal{E}_{\text{PrF}}$ be the set of all PrF pairs  $\left(\avec{X}_M, \avec{Y}_M \right)$ for $N$ and precision  $\epsilon$
\end{itemize}
\begin{lemma}
	For a given $N$ and precision  $\epsilon$
	\begin{equation}\label{eq:erg-sets}
		\mathcal{E}_{\text{OF}} \subset \mathcal{E}_{\text{PF}} \subset \mathcal{E}_{\text{PrF}}
	\end{equation}
\end{lemma}
\begin{proof}
	The relation follows directly from remarks \ref{rem:OFE_PFE} and \ref{rem:PFE_PrFE}: OF is a particular case of PF with $B(y_m, \epsilon_m) = \left\{ y_m  \right\}$ for all $m \in \{ N, \dots, M\}$, and PF  is a particular case of PrF with $P=1$. 
\end{proof}

\section{Probabilistic Forecast  Loci}
Consider the case  where the precision is as in \eqref{eq:pre_epsilon}. 
We can get a locus for the values of $P$ depending on $\epsilon$. An example, to be explained, is shown in Figure \ref{fig:elocus}.
\begin{definition}\label{def:N-locus}
	The \textit{probabilistic forecast  locus} (PrF-locus) $\mathcal{L}$  of a pair  $\left(\avec{X}_M, \avec{Y}_M\right)$ is defined for a given $N$ as the image of the function 
	\begin{equation}
		\mathcal{L}\left(\epsilon \right) = P_\epsilon
	\end{equation}
	over the domain of $\epsilon$,
	where $P_\epsilon$ is the  fraction of the predictions that are precise up to $\epsilon$.
\end{definition}

When the PrF-locus is unknowable, experiments can generate a best effort approximation. In the next section,  Example \ref{ex:PRNG-pure-erg} analyzes the forecast ergodicity of a pseudo random number generator from a theoretical point of view. For that example, the experimental
analysis did not achieve the theoretic optimum. For other problems, the theoretic
optimum might not be known.

An example of the approximation of an PrF-locus is shown in Figure \ref{fig:elocus}. It corresponds to  the results of a particular machine intelligence that learns the pseudo random number generator \cite{random_paper} of Example \ref{ex:PRNG-pure-erg}, yet to be described. 
\begin{figure}[!t]
	\centering
	\includegraphics[width=0.9\linewidth]{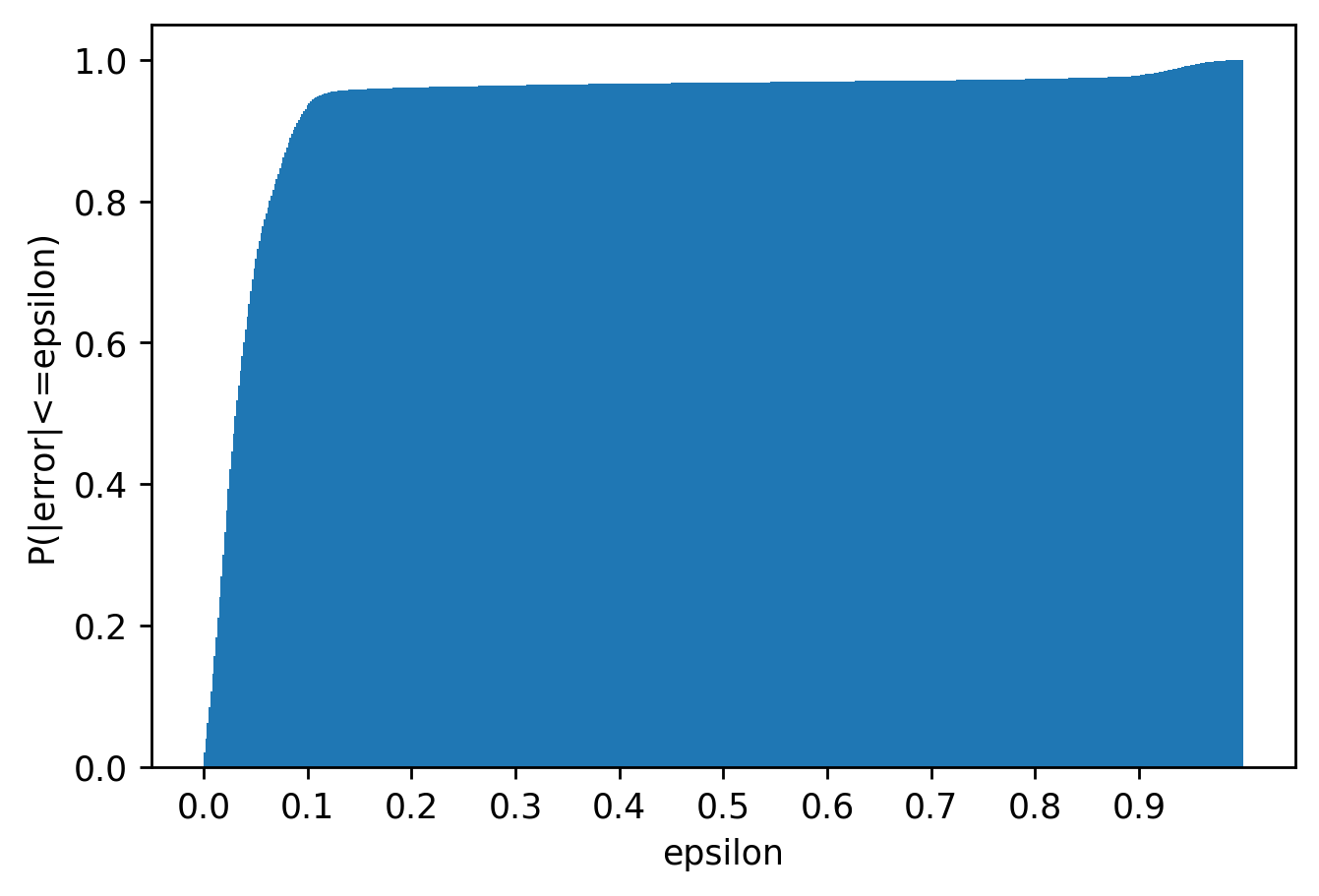}
	\caption{Experimental approximation of the PrF-locus $\mathcal{L}(\epsilon)$ of Example \ref{ex:PRNG-pure-erg}.}
	\label{fig:elocus}
\end{figure}
The plot shown in Figure~\ref{fig:elocus} corresponds to the cumulative distribution function of the probability of the error having magnitude $\epsilon$. In this example, the prediction is always a number in the interval $[0,1)$ and the magnitude of the error is never greater than 1. The predictions were generated by a deep neural network trained with $\num[group-separator={,}]{1000000}$ sequential pseudo random numbers. The graph was generated by measuring the error of the predictions of $\num[group-separator={,}]{100000}$ numbers equally distributed over the interval $[0,1)$ to approximate the PrFE-locus.

\section{Examples}\label{sec:examples}

\begin{example}\label{ex:PRNG-pure-erg}
	
	\textbf{Oracle forecastable (OF)}
	\textit{Pseudo Random Number Generator (PRNG)}\\

	Consider this PRNG \cite{random_paper}
	\begin{equation}\label{eq:PRNG-pure-erg}
		x_{n+1} = \text{frac}\left[\left(x_n + \pi\right)^5\right] .
	\end{equation}
	The PRNG works recursively after an initial seed $x_0$. The expression 
	\begin{equation}\label{eq:PRNG-pure-erg-p} 
		y = \text{frac}\left[\left(x + \pi\right)^5\right] 
	\end{equation}
	is equivalent and it can operate either recursively or in combination with the use of different seeds.
	If $\avec{X}_M$ is the vector of the successive inputs and $\avec{Y}_M$ is the output, then the pair $\left(\avec{X}_M, \avec{Y}_M \right)$ is oracle forecastable for some finite $N$. 
	The mathematical formulation of the PRNG
	\begin{equation}\label{eq:PRNG-pure-erg-f}
		f(x) = \text{frac}\left[\left(x + \pi\right)^5\right] 
	\end{equation}
	is helpful for the proof. Let $x_0$ be a seed for this PRNG and $x_1,\dots,x_{N}$ the outputs of $N$ successive recursive iterations. Then:
	\begin{equation}\label{eq:PRNG-inequality}
		K\left( \avec{Y}_{N} \middle\vert \avec{X}_{N} \right) \lec \ell\left(p\right) + \lg^*\ell(p)
	\end{equation}		
	where  $\avec{X}_{N} = \left( x_0,\dots,x_{N-1} \right)$, $\avec{Y}_{N} = \left( y_0,\dots,y_{N-1} \right) = \left( x_1,\dots,x_{N} \right)$, $\ell(p)$ is the length of a program that implements the PRNG and $\lg^*\ell(p)$ is the length of a self-delimiting code that expresses $\ell(p)$ (see Appendix \ref{apx:log_star}). This is trivially true for any value of $N$, but  equality is needed for the proof. It is not possible to know the value of $N$ that produces
	\begin{equation}\label{eq:PRNG-needed}
		K\left( \avec{Y}_{N} \middle\vert \avec{X}_{N} \right) \eqc K(f)
	\end{equation}
	where $f$ is the function in  \eqref{eq:PRNG-pure-erg-f}, but such $N$ exists and is finite.
	\begin{proof}
		Following the same logic as in Lemma \ref{lem:limitHit} and Corollary \ref{cor:limitHit},  given
		\begin{equation}\label{eq:PRNG-lim}
			\lim_{{M} \rightarrow \infty} K\left( \avec{Y}_{{M}} \middle\vert \avec{X}_{{M}} \right) \eqc K(f)
		\end{equation}
		there exists $N$ such that
		\begin{equation}
			K\left( \avec{Y}_{N}  \middle\vert \avec{X}_{N} \right) \eqc K(f).
		\end{equation}
		Thus
		\begin{equation}
			\begin{split}
				K\left( y_m \middle\vert \avec{X}_{N}, \avec{Y}_{N}, x_m \right) &\eqc K\left( y_m \middle\vert \avec{X}_{N}, f, x_m \right) \\
				&=	K\left( y_m \middle\vert  f, x_m \right) \\
				&\eqc 0
			\end{split}
		\end{equation}
		for all $ m>N$, satisfying the  condition for OFE in  \eqref{eq:ergodic-ai}. 
	\end{proof}
	
	The forecast ergodicity  of the PRNG is:
	\begin{equation}
		\begin{split}
			\mathcal{E}\left(\avec{X}_M, \avec{Y}_M \right) & \le K\left( \avec{Y}_{N} \middle\vert \avec{X}_{N} \right)\\
			& \eqc K(f)
		\end{split}
	\end{equation}
\end{example}

\begin{example}\label{ex:PRNG-precise-erg}
	\textbf{Precise forecastable (PF)}
	\textit{PRNG Web Service} \\
	Let the previous PRNG (Example \ref{ex:PRNG-pure-erg}) be running on a web server with a precision of $s$ bits. It provides numbers with a precision of at least $s'< s$ bits. At any moment, the precision of the output is $s''\in \{s',\dots,s\}$ bits. Let a sensor activated by radioactive decay  trigger a change of precision, so the exact precision at a given time is unknown.
	If $\avec{X}_M$ are the inputs to the PRNG on the server, with a precision of $s$ bits, and $\avec{Y}_M$ are the outputs on the clients, with a precision of at least $s'$ bits, then the pair $\left(\avec{X}_M, \avec{Y}_M \right)$ is PF for $\epsilon = 2^{s-s'}$. 
	
	Regarding the analysis in  \eqref{eq:ce}, the  KCS complexity of this noise can be bounded with $s-s'$ bits plus the length of a self-encoding expression for $s-s'$ (see Appendix \ref{apx:log_star}),
	$$ C_\epsilon \leq s-s' + \lg^*\left(s-s'\right) +c. $$
	up to a constant for the sign of $s-s'$.
\end{example}

\begin{example}\label{ex:PRNG-prob-erg}
	\textbf{Probabilistic forecastable  (PrF)}
	\textit{Signal Spectrum Neural Network Classifier}\\ 
	Consider a neural network (NN) binary classifier trained to distinguish between  LTE and Bluetooth signals. Bluetooth band (2400 to 2482.25 MHz) lies in between adjacent LTE bands, which are very close in the lower end and a little more separated on the upper end. Leakage between bands is very probable and other interferes, such as  harmonics interference caused by the multiples of LTE frequencies in the Bluetooth band, can mess up the signals and make the classifier fail. In this context, the pair  $\left(\avec{X}_M, \avec{Y}_M \right)$ is, presumably, PrF, where the inputs are the signals and the outputs are the classifications. Let assume that the training data cover many typical cases that is possible to classify, while $1\%$ of the future signals  are not represented by the training data. Then, the fraction $P$ of the future outputs that can be predicted to a certain precision could potentially be up to $99\%$. 
	Experimentally,  validation data\footnote{The validation data consist of a subset of the available data at the moment of training. It is reserved apart from the training data to evaluate the training process. It provides an independent set of data to check how the system generalizes outside of the training data.} 
	serves to estimate $P$. 
	Let the output  $y_m\in [0,1]$ be within the interval $[0,0.25]$ for Bluetooth and $[0.75,1]$ for LTE, so we impose a precision of $\epsilon = 0.25$. For an actual $P=0.99$, the experimental value will most probably be lower, for instance $P=0.95$, if $95\%$ of the validation data is correctly classified. The experimental value of $P$ is expected to be  smaller than the theoretical one, but if the validation data is biased, the experimental $P$ can appear better. In any case, the validation set provides the  best experimental approximation to the actual value.
\end{example}

\begin{example}\label{ex:prfeReinfLearn}
	\textbf{PrF}
	\textit{Reinforcement learning}\\ 
	A deep learning system that uses reinforcement learning updates the parameters, $\Theta$, of an artificial neural network model (agent)   over time. At each  iteration, $m$, the output of the agent is an action, $A_m$, and the inputs to the agent   are the state of the environment, $S_m$, and a reward, $R_m$.  The action can affect the state of the environment, which gives a reward to the agent. For this example, let
	\begin{align*}
		x_m &= \left(\Theta_{m-1}, S_m, R_m \right),\\
		y_m &= \left(A_m, \Theta_m\right).
	\end{align*}
	The goal is to generate and maintain  an agent whose actions  optimize the return (aggregation of rewards). The reward can be absent or delayed and it can have a stochastic ingredient, so, the pair $\left(\avec{X}_M, \avec{Y}_M \right)$ is PrF. 
\end{example}

Sutton and Barto \cite{Sutton1998} note that most modern reinforcement learning uses finite Markov decision processes. 

\begin{example}
	\textbf{OF}
	\textit{Finite Markov reinforcement learning}\\ 
	Adding some constrains to a reinforcement learning problem using a finite Markov decision process, we can obtain an OF problem. Consider the case where both the reward and the next state are a function of the action and the previous state, and the goal is to train an agent that chooses the action that optimizes the reward. This can apply to games such as checkers, chess, or go. These classes of games are OF for a finite $N$ \cite{checkers}.
\end{example}

The next example is an application of forecast ergodicity to dyadic numbers. Based on \cite{bil1995prob}, let $T:(0,1] \rightarrow(0,1]$ be a mapping from the unit interval into itself defined by
\begin{equation}
	T_\omega = 
	\begin{cases}
		2\omega & \text{if } 0<\omega\le \frac{1}{2} \\
		2\omega - 1 & \text{if } \frac{1}{2} < \omega \le 1.
	\end{cases}
\end{equation}
Define also a Boolean function $d_1$ on $(0,1]$ by
\begin{equation}
	d_1 = \begin{cases}
		0 & \text{if } 0<\omega\le \frac{1}{2} \\
		1 & \text{if } \frac{1}{2} < \omega \le 1,
	\end{cases}
\end{equation}
and let $d_i(\omega) = d_1(T^{i-1}\omega)$. Then
\begin{lemma}
	\begin{equation}
		\sum_{i=1}^{n}\frac{d_i(\omega)}{2^i} < \omega \le \sum_{i=1}^{n}\frac{d_i(\omega)}{2^i} + \frac{1}{2^n},
	\end{equation}
for all $\omega \in (0,1]$ and $n \le 1$.
\end{lemma}
\begin{proof}
	Easily proven by induction considering separate cases for $\omega \le \frac{1}{2}$ and $\omega>\frac{1}{2}$ (see \cite{bil1995prob}).
\end{proof}

Every real number $\omega \in (0,1]$ can be arbitrarily approximated by dyadic numbers, which are always rational. In fact, for $\epsilon_n = 2^{-n}$, each ball $B(\omega, \epsilon_n)$ contains $\sum_{i+1}^{n}\frac{d_i(\omega)}{2^i}$, the dyadic approximation up to $n$ terms.

\begin{example}\label{ex:dyadicRational}
	\textbf{PF}
	\textit{Dyadic approximation of real numbers}\\
	Every real number $\omega \in (0,1]$ can be arbitrarily approximated by dyadic numbers. Let  each $x_n \in \left\{\avec{X}_\infty\right\}$  be the dyadic expansion of $\omega$ with $n$ digits of precision, and for every $y_n \in \left\{\avec{y}_M\right\}, y_n = \sum_{i+1}^{n}\frac{d_i(\omega)}{2^i}$. For instance, if $\omega = 0.01011$, then $x_1 = 0.0, x_2 = 0.01, x_3=0.010, x_4 = 0.0101, x_5 = 0.01011$, and $y_1=0, y_2 = 0.01, y_3=0.01, y_4=0.0101, y_5=0.01011$. 
	For all $m \in \mathbb{N}$ define a precision $\epsilon_m = 2^{-m}$ for the ball $	B(\omega, \epsilon_m) = \left\{  \omega ' \in \mathbb{R} : \left| \omega -  \omega' \right| < \epsilon_m \right\}$. Then, for all $\epsilon > 0$ there exists $N$ such that the output is precise forecastable. 
\end{example}

\section{Conclusion}

Machine intelligence can accumulate data from
the past to the present and do predictions within the boundaries of the algorithmic information forecastability of
the data. Using algorithmic information theory, oracle, precise and probabilistic forecastability have been defined and illustrated through examples. 

Forecastability is a property of the data and not the method used to perform the prediction. Not all data can be forecast, the most obvious of which is the repeated flipping of a fair coin.   Forecastability does not necessarily mean the forecasting is straightforward or computationally inexpensive, e.g. forecasting the output of certain pseudo-random number generators. Design of forecasting algorithms  remains a function of domain expertise about the data and computational power.

Marvin Minsky opined \cite{minski}  ``everybody should learn all about that [algorithmic information theory forecasting] and spend the rest of their lives working on it.'' This paper is a step in application of algorithmic information processing to Minsky's challenge.

%
%

{\appendix[Recursive Self-Delimiting Code for the Function $\text{lg}^*$]\label{apx:log_star}
	To run a program $p$ in an universal computer $\mathcal{U}$, an indication of the length of the program has to be provided so it can halt when the end of the program is reached. And, in order to indicate the length of the program, the computer $\mathcal{U}$ needs to be informed about when  to stop reading that number. Self-delimiting codes are used for this purpose.

For instance, 00 could encode the value 0 and 11 the value 1, while 01 would indicate the end of the number. The cost of this approach is twice the length of the number plus 2 bits for the ending: $2\log(n)+2$ bits.

A better approach, introduced in  \cite{cover} and \cite{vitanyi}, consist of concatenating the length, with the length of the length, with the length of the length of the length, and so on, recursively, by adding the concatenations on the left, until the value 1 is reached. The self-encoding of the number $n$ would be something like this:
\begin{equation*}
	1 \mathbin\Vert \dots
	\mathbin\Vert \log\left(\log\left(\log\left(n\right)\right)\right) 
	\mathbin\Vert \log\left(\log\left(n\right)\right)
	\mathbin\Vert \log\left(n\right)
	\mathbin\Vert n
	\mathbin\Vert 0
\end{equation*}
where $\mathbin\Vert$ is the symbol for the concatenation of two strings, and the 0 at the right indicates the end of the encoding.

In \cite{vitanyi}, when presenting this method, a previous knowledge of the depth of the  recursion is assumed. However, adding some simple rules allows for a fully self-delimiting implementation without indication of the depth of the recursion. For example, the implementation we developed, based on this approach, encodes the number 1200 as
\begin{verbatim}
	1111001011100101100000
\end{verbatim}
with a length of 22 bits. It decomposes as
\begin{verbatim}
	                     0 -> end
	          10010110000  -> 1200 in binary
	      1011             -> length of 1200
	   100                 -> length of 11
	 11                    -> length of 4
	1                      -> start
\end{verbatim}
which is written starting from the right (first 0, then the number, then the $\log$ of the number, then the $\log\log$ of the number, and so on), but it reads from bottom to top as in the next explanations
\begin{lstlisting}[breaklines]
1 -> start (will read 2 next bits when
     the code starts with a 1)
 11 -> 3 in binary. Will read next 3 bits
       (the binary length of 4 is 3):
       ceil(log2(4+1)) = 3
   100 -> 4 in binary. Will read next 4 
          bits(the binary length of 11
          is 4): ceil(log2(11+1)) = 4
      1011 -> 11 in binary. Will read
              next 11 bits (the binary 
              length of 1200 is 11): 
              ceil(log2(1200+1)) = 11
          10010110000 -> 1200 in binary
                     0 -> end
\end{lstlisting}
The encoding  process goes from right to left. 
This implementation allows the computer to read the number and detect the end, with the rule of stopping when the next lecture starts by 0.  With access to the decoding algorithm, no  previous knowledge about the number or about the depth of the recursion is needed in order to decode it. The implementation used for this example encodes and decodes integers greater than or equal to 0 (the sign is provided independently). Even when some details are obviated from this explanation ---details such as the special cases of the numbers 0 and 1 or the exceptions for the first iteration--- the example provides a clear view of how the method works.
For the decoding process, the reading starts from the left to the right as follows:
 after processing the first starting bit 1, then it reads $11_2 = 3$, so it now must proceed to read the next 3 bits, that is $100_2 = 4$, then it reads the next 4 bits, that is $1011_2 = 11$, it reads the next 11 bits, that is $10010110000_2 = 1200$, then, the next bit is 0, indicating the end of the encoding. So, the last number in the concatenation, 1200, is the self-delimited encoded number.

This self-delimiting encoding of the number 1200 has length 22; that is only a little  shorter than the bit duplication  method for this number, which requires $2\times11+2=24$ bits. However, for larger numbers, the optimization becomes more clear, as shown in Figure \ref{fig:graphlgstar}.
\begin{figure}[!t]
	\centering
	\includegraphics[width=0.9\linewidth]{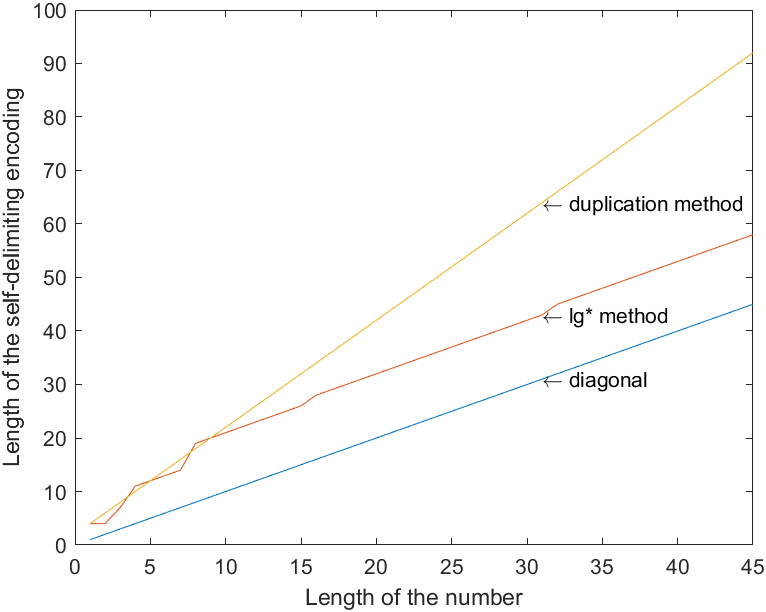}
	\caption{Length in bits of the self-delimiting encoding method of bit duplication and $\lg^*$ as a function of the length of the binary numbers. The difference between the bit duplication and the $\lg^*$ method is significant when the numbers are  large. }
	\label{fig:graphlgstar}
\end{figure}
Taking as reference the length of the number (diagonal line on Figure \ref{fig:graphlgstar}) the additional amount of bits required for bit duplication grows as $O(n)$, while with the $\lg^*$ method requires less bits for the representation.

}

\bibliographystyle{IEEEtran}
\bibliography{biblio}

\vfill

\end{document}